\documentclass[conference]{IEEEtran}
\IEEEoverridecommandlockouts
\usepackage{cite}
\usepackage{amsmath,amssymb,amsfonts}
\usepackage{algorithm}
\usepackage{graphicx}
\usepackage{textcomp}
\usepackage{tcolorbox}
\usepackage{xspace}
\tcbuselibrary{skins, breakable, theorems}
\usepackage{flushend}
\usepackage{amsthm}
\usepackage{subcaption}
\usepackage{algorithm}
\usepackage{algcompatible}
\usepackage{comment} 
\usepackage{xcolor}
\newcommand{\stef}[1]{{\color{blue}{#1}}}
\newcommand{\protocol}{TrustMix\xspace}
\def\BibTeX{{\rm B\kern-.05em{\sc i\kern-.025em b}\kern-.08em
    T\kern-.1667em\lower.7ex\hbox{E}\kern-.125emX}}
\newtheorem{theorem}{Theorem}
\newtheorem{definition}{Definition}
\begin{document}

\title{TrustMix: How to Mix Messages in a Mobile Ad-hoc Network}

\author{\IEEEauthorblockN{Yu Shen, Aiswarya Walter, Stefanie Roos}
\IEEEauthorblockA{\textit{} \\
\textit{RPTU University Kaiserslautern-Landau}\\
Kaiserslautern, Germany \\
\{yu.shen, stefanie.roos\}@cs.rptu.de, a.walter@edu.rptu.de}
}
\maketitle
\begin{abstract}

Mix networks are a highly effective way to achieve anonymity, defending against a wide range of traffic-analysis attacks.
However, mix networks are usually designed for infrastructure networks and cannot be directly applied in the context of mobile ad hoc networks (MANETs). The few existing solutions for MANETs require advance knowledge of the topology or a trusted central party.

In this paper, we present TrustMix, a mix protocol for MANETs that operates without any central trusted party. In TrustMix, parties join groups and then messages are forwarded via multiple groups to provide anonymity. With TrustMix, users only need to find a party nearby that they consider trusted. They then forward the message to  this party's group, and the party shuffles messages before forwarding to other groups, meaning that the original message and the forwarded message cannot be linked. Furthermore, even if the chosen party is adversarial, they can only break the anonymity if all parties in their group are adversarial as all of them contribute to the shuffling. In addition to anonymity, TrustMix also enforces rate limits on the number of messages through the use of linkable ring signatures, which allows detecting that parties send more messages that allowed without revealing identities. 

We prove the security of our protocol in the random oracle model. We evaluate its anonymity using an existing mix-network simulator and show that TrustMix significantly improves message anonymity. Finally, we present a proof-of-concept Android implementation and show that TrustMix achieves acceptable throughput with 5 mobile devices.
\end{abstract}


\section{Introduction}

Instant messaging applications such as WhatsApp and Signal are widely used to communicate, read the news, or even do business~\cite{protestCommunication}. Most  applications require internet access, which may in certain scenarios not be available, e.g., because of an internet shutdown. For example, authorities may block the Internet in a region during a protest, as has occurred in Russia and Iran~\cite{iranInternetShutdown}. Then, most communication tools are not available, and it becomes extremely difficult to communicate with friends or relatives and to access or disseminate information. 

There already exist tools such as Firechat and Bridgefy that do not require internet access~\cite{Bridgefy}. Instead, they use low-range radio techniques such as Bluetooth and WIFI Direct to transfer messages. When a sender sends a message, they simply broadcast the message to all of their neighbors, and other users will broadcast the message further until all users receive this message. The broadcast of messages ensures that the receiver always gets the message if there is at least one path between the sender and the receiver. However, a simple broadcast opens the door for denial-of-service (DoS) attacks, in particular allowing adversaries to flood the network with dummy messages, meaning that a limited number of malicious parties can easily paralyze the network. Moby and Rangzen mitigate DoS attacks by introducing the use of a trust mechanism~\cite{moby,rangzen}. With the trust mechanism, only messages from trusted users are forwarded, with the expectation that messages sent by malicious parties end up being dropped before they reach a considerable fraction of the network due to the low trust placed in the parties. 

However, none of the protocols, not even Moby and Rangzen, provide anonymity in the face of a global passive adversary, as they simply encrypt the message and assume that no one except the receiver can obtain the plaintext message. In such systems, the global passive adversary can easily trace every encrypted message back to the original sender through traffic analysis, enabling the adversary to derive relationships between users, which can reveal sensitive information such as political affiliation.

Since the introduction of mix networks by Chaum~\cite{ChaumProof92}, they have been widely used in distributed systems to increase anonymity~\cite{nym}. The core idea of a mix network is to have a set of mix nodes collect, transform, and shuffle messages to break the link between senders and receivers. The original design does not suit the nature of MANETs because it assumes fixed topologies and routing strategies. Existing mix-inspired MANETs communication protocols such as ARM and ANODR~\cite{arm,anodr} attempt to handle network dynamics through route discovery and by forwarding messages along pre-established paths. In practice, however, intermediate nodes may go offline at any time and network topology can change rapidly, making such route-based designs unsuitable for highly dynamic MANETs. Moreover, the request messages and reply messages to discover routes make them vulnerable to denial of service attacks, as adversaries can flood the whole network with route discovery messages.

Motivated by these limitations, we propose \protocol, a new anonymous communication protocol that combines trust-based forwarding with decentralized mix networks to provide anonymity and DoS resistance under dynamic network conditions.
We introduce a trust mechanism into \protocol to mitigate the impact of DoS attacks similar to the Moby network~\cite{moby}. In \protocol, messages are only exchanged among trusted parties. Hence, messages from malicious parties are dropped with high probability. To ensure the anonymity of the messages, we integrate a mix network protocol in \protocol. To overcome the limitations of a dynamic network topology, we use a similar idea as Atom~\cite{Atom} to build a free route mix network. In Atom, a mix node consists of a group of parties that collaboratively re-encrypt and shuffle messages,  meaning messages are forwarded to arbitrary next hops without revealing plaintexts. 
Similarly, in \protocol, users self-organize into groups, and each group collectively acts as a mix node. Senders only select the first group, while subsequent routing decisions are made by intermediary groups. Therefore, users do not need to determine every mix node and messages are resilient to faulty mix nodes. To select groups, each party has a set of other parties they consider trusted and forward to a group with at least one party they trust. This trusted party then assumes a key role in the re-encryption of the message and they choose the next group to forward to, though even if the trust is misplaced, \protocol provides anonymity as long as a sufficient number of group members are honest.

We prove the security of \protocol and evaluate its anonymity using a simulator, finding that the anonymity provided can be up to 8 to 10 bits. In addition to this, we implement a prototype of \protocol on Android devices and benchmark its delay and throughput, showing that the latency increases approximately linearly in the number of group members. 

In short, \protocol has the following contributions:
\begin{itemize}
\item We design an anonymous trust discovery protocol that enables DoS mitigation without revealing the relationship of users.
\item We design the first trust-based re-encryption mix protocol for MANETs.
\item We evaluate the anonymity of \protocol using a simulator. 
\item We implement an Android WIFI Aware based prototype and evaluate the performance of \protocol.   
\end{itemize}

\section{Related works}
In this section, we introduce the works most closely related to \protocol, including three areas: mix networks, anonymous communication protocols for mobile ad hoc networks, and trust mechanisms for decentralized systems.
\subsection{Mix networks}
Since the proposal of mix networks by Chaum~\cite{ChaumProof92}, they have become a fundamental architecture for anonymous communication. In Chaum’s original design, the sender selects a sequence of mix nodes and applies layered encryption. Afterwards, each mix node collects incoming messages, decrypts them, and forwards them to the next mix node or the receiver, thereby breaking the linkability between senders and receivers of messages. Vuvuzela~\cite{Vuvuzela} employs a single chain of mix nodes combined with heavy cover traffic to provide strong privacy for message metadata. Similarly, Loopix~\cite{Loopix} uses continuous cover traffic 
but arranges mixes in layers so that each route contains one mix from each layer. 

All systems above require that the topology of the mix network be known to users when sending a message, and changes of the topology or availability of mix nodes may lead to protocol failure. In an MANET with dynamic topology, it is difficult for users to obtain the network topology on time.  

In contrast to sender-routed mix networks, Atom introduces the idea of a re-encryption-based mix network where intermediaries are able to re-encrypt a message with the public key of an arbitrary party and forward the message to this party~\cite{Atom}. Atom also organizes multiple nodes into a single mix group, allowing correct mixing of messages even when some nodes are unavailable. These properties make Atom suitable for MANETs due to its fault tolerance and flexible message routing method. However, Atom relies on reliable group assignment by a trusted party, an assumption that is difficult to realize in a decentralized MANET.

\subsection{Anonymous MANETs}
Due to the flexibility of mobile ad hoc networks, MANET-based applications such as Bridgefy have become common tools for communication during Internet shutdowns or blackouts~\cite{Bridgefy}. To provide anonymity in such environments, several protocols have been proposed with different design goals. Rangzen is designed to broadcast a message anonymously in a MANET with potential adversaries that attempt to paralyze the network~\cite{rangzen}. Rangzen uses a prioritization mechanism, in which messages from parties that share more mutual contacts with the current node have a higher priority. In such a way, messages from adversaries likely have a lower priority and may be dropped by intermediary parties under congestion. Similarly to Rangzen, Moby performs anonymous peer to peer communication in a MANET using a trust based prioritization mechanism~\cite{moby}. Amigo further expands the use cases of MANETs by supporting secure group communication~\cite{amigo}, supporting dynamic group creation, membership management, and anonymous messaging over MANETs.
However, all of the anonymous communication protocols discussed above assume that the adversary cannot monitor the traffic of the entire network. Consequently, these protocols remain vulnerable to traffic analysis attacks, motivating the need for systems that provide stronger message anonymity.

\subsection{Trust mechanisms for decentralized systems}
Trust mechanisms are widely studied in decentralized systems to mitigate active attacks without centralized authorities. As mentioned in the previous section, both Moby and Rangzen use a trust mechanism where each party assigns a trust score to each message and prioritizes messages according to trust values. In Rangzen, message receivers calculate the trust value of a message based on the number of social contacts shared between them and the message sender. In contrast, the trust value of a message in Moby is updated at every forwarding step, where the new trust value is calculated based on the previous trust value of the message and the trust value of the forwarder.

Anix introduces a method to remotely establish trust relationships in MANETs over several rounds of communication\cite{anix}. Concretely, each user generates a new pseudonym every time they send a message. When a user receives a trustworthy message together with the public key corresponding to the sender’s pseudonym, the receiver can choose to initiate the trust establishment process by broadcasting their identifier encrypted by the public key, allowing the original sender to track subsequent messages. If choosing to establish mutual rust, the original sender broadcasts the encrypted form of their own identifier in return, leading to the establishment of two-way trust. In \protocol, we adopt a trust establishment process similar to Anix to find trusted parties near each other without revealing the identifier of each participant.

\section{Threat model}
We consider a mobile ad hoc network during internet shutdowns during which devices communicate without fixed infrastructure and dynamically organize into groups that collaboratively perform message mixing and forwarding. The network topology is highly dynamic and communication links are unreliable.  
We focus on three adversarial goals: denial-of-service, modifying the content of messages, and linking senders and receivers to break anonymity. 

For denial-of-service attacks, attacks that completely prevent communication, physical capture of all devices, and global radio jamming are out of scope. Rather, we focus on local, internal attacks where malicious parties flood the network or drop messages.

For the anonymity and message integrity, we adopt a threat model similar to that of Moby~\cite{moby} but with a stronger global passive adversary. We assume that a probabilistic polynomial-time (PPT) adversary can corrupt a subset of nodes and fully control their behavior, but cannot break standard cryptographic primitives. With corrupted nodes, the adversary can read the internal state of the corrupted nodes, modify their behavior, inject, drop, delay, or flood messages, and collude across all compromised devices. In addition, the adversary can observe the traffic patterns of all communication links. Although monitoring of all channels may not be feasible, we provide such information to adversaries to provide strong anonymity guarantees for \protocol.

%During blackouts, we do not consider the use of the Internet as it may be blocked by communication providers. Hence, users of \protocol only use short-range radio based technologies such as WIFI Direct and Bluetooth to communicate with each other. We assume the existence of some trust materials \stef{weird term} such as contact books and contact histories \stef{do you trust everyone in your contact book or contact history? requires more motivation} so users are able to infer whether a party is trustful. 

%Depending on the trust mechanism\stef{is that a mechanism or an assumption? I am not sure what about the way we define trust is a mechanism?}, the risk of trusting a user may vary. With a direct trust mechanism where users only trust people they know in reality, the trusted party is honest with a high probability. For some indirect trust mechanisms where users trust unknown people, the trusted party may not be reliable \stef{is there a difference between reliable and trusted?}. In this paper, we assume the trusted party always behaves honestly because studying how to build a reliable trust mechanism is out of the scope of this paper and we simply assume a reliable trust mechanism is used. \stef{that seems like a huge problem, normally you want a protocol to fail gracefully if your assumptions are not met}

\section{Protocol}
\label{sec:protocol}
The goal of \protocol is to anonymously mix messages in decentralized mobile ad hoc networks. Unlike Atom, which assigns fixed groups beforehand, we allow users to form a group, leave a group, and join a group during protocol execution, meaning that \protocol can adapt to topology changes. Moreover, we consider the occurrence of failures of users by allowing groups to process incoming messages in the presence of several non-responsive members. For $k$ non-responsive users, Atom requires at least $k+1$ honest users to ensure the message can be processed correctly. In our protocol, we only require one honest party in the group, regardless of the number of non-responsive users. However, the honest party needs to be chosen and trusted by the previous message forwarders in our protocol. We ensure that if the selected trusted user is honest, incoming messages of all group members are shuffled securely, and no one can link shuffled messages to original messages with a non-negligible probability. To prevent forging and modification of messages by group members as well as DoS attacks, a misbehavior detection mechanism is integrated to find misbehaving parties in our protocol. 

%Note that just using trusted forwarders instead of groups is not sufficient as it cannot:
%\begin{itemize}
    %\item keep the content of messages secret: the trusted party learns the content of a message when it decrypts an incoming message, while our protocol keeps the content secret during the shuffling and re-encryption. 
    %\item collect enough messages to mix: group members shuffle messages from all group members in our protocol, but in the naive method, the trusted party can only receive messages from its neighbors unless a multi-hop routing algorithm is added to that method.
%\end{itemize}

Our protocol consists of a setup phase and a mixing phase. The setup can be executed by any set of users to create a set of group public keys for a new group, such that each group member is associated with a group public key. Later, if a party wants to forward a message to this group and a member of this group is trusted by the forwarder, the forwarder can choose the member as the next trusted forwarder and encrypt the message under the corresponding group public key.

In a group, all incoming messages are broadcast to every party in the group, including messages for other groups' group public keys and messages to be mixed. When a new message arrives, group members store it in a message pool until there are enough messages to be mixed. If the message pool is full, the message mix phase is performed. 
The mixing phase has 6 steps:

\begin{itemize}
\item  DoS-resistant incoming message handling: for each incoming message, the previous forwarder appends its signature. This signature is generated using a trusted public key of a member of the current group. Members of the current group verify the signature and broadcast the incoming message if the signature is valid. Otherwise, the incoming message is dropped. If the number of messages signed by the same trusted public key exceeds a limit, subsequent messages signed by the same key are discarded.

\item  Verifiable shuffling: in this step, incoming messages are shuffled to increase the anonymity of messages. Initially, all incoming messages are stored in a list. For each group member, an identical message list is maintained since every message is broadcast to every member. Then, every group member initializes the shuffling of the message list using the group public key assigned to it.

\item  Decryption and check: in the previous step, several messages were shuffled using a public key different from the public key used for encryption. In this step, users decrypt a special message to check and remove messages shuffled using a wrong public key.

\item Anonymous public key re-encryption: after shuffling of messages, the trusted forwarder needs to decide the next trusted forwarder and re-encrypt the message for the next trusted forwarder. The current trusted forwarder asks other group members to re-encrypt the message using an anonymous group public key that is created from the group public key of the next trusted forwarder and a randomly sampled public key to ensure other members are not able to learn who the next trusted forwarder is.

\item Decryption with the anonymous secret key: to remove one layer of encryption added in the previous step, the trusted group member decrypts the re-encrypted message using the corresponding secret key of the randomly sampled public key. Afterwards, the incoming message is encrypted using the group public key of the next trusted forwarder anonymously.

\item Forwarding: group members forward messages to a group decided by the trusted forwarder. 

\end{itemize}

\subsection{Setup}
In the setup phase, we construct one distinct group public key for every group member. We use a Secure Distributed Key Generation scheme to ensure
1. a message can be processed only with the participation of the trusted forwarder and at least $n_{min}-1$ other group members where $n_{min}$ is the minimum number of required participants 
2. a message can be processed when $n_{}-n_{min}$ group members are unresponsive and the trusted forwarder remains responsive where $n_{}$ is the number of members in a group.

To generate a group public key for the $i$-th member of a group, we use the Secure Distributed Key Generation protocol proposed by Gennaro et al.~\cite{dkgprotocol} with threshold $(n_{},2n_{}-n_{min})$ to generate the public key and partial secret keys of the group. After key generation, the $i$-th member holds $n_{}+1-n_{min}$ key pairs for the $i$-th group public key $gpk_i$, while other members hold exactly one pair. Consequently, messages encrypted under $gpk_i$ can be decrypted only by combining the secret keys of the $i-th$ member with at least additional $n_{min}-1$ secret keys.

The group public key generation process is executed periodically, so new users can be included continuously after a waiting time. Members only need to adjust the $n_{}$ value according to the number of members in the group.

\subsection{Anonymous trusted party discovery and DoS prevention}
In this section, we present the design of our anonymous trusted party discovery protocol. It enables users to send messages to nearby trusted parties, while keeping trust relationships secret. We also introduce a trust-based DoS-resistant incoming message handling mechanism in this section.

An advantage of Moby is that it can mitigate DoS attacks by only forwarding and receiving messages from trusted parties. To prevent DoS attacks from other groups, we let group members only receive messages from trusted parties and forward messages only to trusted parties. To achieve this goal, we need a trusted party discovery protocol to let group members know which party in neighboring groups is trusted. A straightforward approach would be to broadcast the public keys of trusted parties. However, this method reveals the trust relationship between users that can compromise privacy and break message unlinkability. For example, if a message is sent from a trusted party of a specific group member and this group member forwards a message later, adversaries are able to link these two messages if the trust relationship between the group member and the message sender is known. 

To avoid this issue, we design an anonymous trusted party discovery protocol to protect the unlinkability of messages. Similarly to Anix\cite{anix}, we assume that there is a trust public key $PK_i$ for each user, and such keys are only known by trusted parties of a user. As the first step in trusted party discovery, each party generates a pseudonym $nym$ every time it joins a group. This pseudonym is essentially a randomly generated public key $PK_{nym}$ and the corresponding private key $SK_{nym}$ is kept secret. To inform trusted parties about a user, this user signs $PK_{nym}$ using its trust secret key $SK_i$. So, trusted parties of this user know that the anonymous public key $PK_{nym}$ is generated by one of its trusted users. We use the key blind signature scheme proposed by Denis et al.~\cite{blindsign} similar to Anix to ensure that other parties are not able to gain knowledge about $PK_{i}$ based on the signatures generated using $SK_{i}$. Afterwards, the signed $PK_{nym}$ acts as the anonymous identifier of the user and is broadcast to all neighboring groups. To prevent DoS attacks, such identification messages also need to be signed collaboratively by group members. Therefore, the identification message of every group member is aggregated and signed by a group public key. Unlike the group public key owned by different members, this group public key is not generated for a specific trusted group mechanism, and each member has the same number of shares for the corresponding private key. In such a way, a group only needs to sign and broadcast the identification message once unless there are new members joining this group. For the corresponding trusted parties in neighboring groups who know the public key, they can verify the signature and know $PK_{nym}$ of the trusted parties in the received identification message.

As mentioned above, group members only process messages signed by a known trusted public key. So, group members need a way to add trusted parties' keys to an identical list maintained by all group members, which we refer to as the trust list, without losing anonymity. $PK_{nym}$ cannot be used directly because it is a long-term key that potentially multiple parties, namely all trusted parties of the corresponding user, can link to the real identity, allowing adversaries to infer relationships through long-term observations. Instead of $PK_{nym}$, we let a party calculate an anonymized trust key $TPK=PK_{1}^{Hash({PK_{2}}||{PK_{nym}})}$ where $PK_{1}$ is the trust public key of the trusted party in the neighbor group, $PK_{2}$ is the trust public key of the user itself, $PK_{nym}$ is the pseudonym of the user, and $Hash$ maps a string to a scalar of the underlying group. The trust key now combines the public key of both parties and also a randomly generated pseudonym. Thus, different users have different trust keys even if they trust the same party. Also, if a user joins a new group, they agree on a new trust key with all of their trusted parties, as $PK_{nym}$ is group-dependent. As a result, unless an adversary knows both $PK_1$ and $PK_2$, it cannot infer trust relationships with non-negligible probability.

After the calculation of anonymized trust keys, group members aggregate their anonymized trust keys using the multi-party private set union (MPSU) protocol proposed by Gao et al.\cite{privateSet} to hide the owner of each anonymized trust key. Such a PSU operation has an upper limit on the set size, so malicious parties are not able to apply DoS attacks by adding a large number of trusted parties. With PSU operations, adversaries cannot link an anonymized trust key to a specific member. The owner of $PK_1$ can calculate the secret key as $TSK = SK_{1}\cdot{Hash({PK_{2}}||{PK_{nym}})}$ where $SK_{1}$ is its own secret key, $PK_{nym}$ is the pseudonym it has received and verified before, and $PK_{2}$ is the corresponding trust public key for the pseudonym, considering ElGamal encryption scheme is used. Later, if it wants to forward a message to the group of the trusted party, it can sign the message using $TSK$ because $TPK$ is added to the neighbor group trust list if the trusted party executes the trust discovery protocol correctly.

To realize a \textbf{DoS-resistant incoming message handling}, all messages except identification messages must be signed with an anonymized trust key. Accordingly, group members process an incoming message only if it carries a valid signature under a key contained in the group’s trust list.
However, using signatures under a fixed anonymized trust key would allow any group member to link multiple incoming messages, since all such signatures are publicly verifiable under the same key. To provide anonymity for signers while still enabling abuse detection, we employ a linkable ring signature (LRS) scheme~\cite{linkableRing}. To sign messages with LRS, the sender of messages needs to know the ring to be used beforehand. Therefore, the trust list is also broadcast to neighboring groups after the anonymized trust keys are aggregated from the group members. Then, the sender can sign a message using $TSK$ it holds and the ring of a group. LRS supports signatures to a specific "event-id". We use the current timestamp as the "event-id". Hence, signatures generated using the same $TSK$ under the same event-id are linkable to each other by the linkability of LRS, and corresponding messages can be identified and dropped by group members. Consequently, a party cannot flood a group with dummy messages. 

The trust mechanism alone is not sufficient to resist traffic analysis, since an adversary can link trusted forwarders with potential senders by observing traffic patterns. In particular, consider a scenario in which several users trust the same group member and forward their messages to that party. Later, the trusted group member will handle all those messages, while other members do not have messages to process. In this case, the adversary can learn that the senders of these messages trust a specific group member. For subsequent messages from the same senders, the adversary can link the outgoing messages of the corresponding trusted group member to messages sent from those senders with a non-negligible probability, which weakens the unlinkability guarantees of \protocol. To avoid this problem, we require potential senders to transmit dummy messages when they have no real messages to send. Assuming an anonymized trust key can be used to sign up to $k$ messages per minute without being linkable, and the owner of this key only has $a$ messages to send, it generates $k-a$ dummy messages and forwards them to the corresponding trusted party. 
In practice, a sender sends $k$ dummy messages every minute and stores messages to be sent in a list. When a dummy message is to be sent, the dummy message is dropped and replaced with the first message in the list if the list is not empty. 
As a result, group members exhibit a uniform outgoing traffic pattern over time, preventing an adversary from inferring relationships between message senders and group members based on the traffic pattern.

%The last issue to tackle is detection of lost or dropped messages, e.g., to enable retransmissions: We use acknowledgments to have the recipient confirm that a message has been delivered successfully. On arrival of a message, the receiver checks whether this message is signed by the key in their trust list. If so, it broadcasts this message to other members. Afterwards, group members sign the hash of the message collaboratively and forward the signature back to the nearby groups periodically. In such a way, group members are able to resend the message in case of message loss or dropping by adversaries. \stef{so there are timeouts involved? how are they chosen?}

\subsection{TrustMix}
%Every time a group of users wants to construct a group or a group wants to add a new user, they run the public key initialization function of the weighted threshold to ensure there is a group public key $gpk_i$ for each group member $u_i$. For $gpk_i$, $u_i$ has the highest weight since it is the public key to be used when $i-th$ member is trusted. 
After the anonymous trusted party discovery, group public keys and identification information are broadcast to neighboring groups so neighbors learn about trusted parties nearby and their group public keys. When a sender sends a message, it randomly chooses a party from all its nearby trusted parties as the trusted forwarder. The sender then encrypts its message under the group public key of the chosen trusted forwarder, signs this message using the corresponding $TSK$, and forwards the resulting ciphertext to the group of the trusted forwarder

When receiving messages, parties wait for enough messages similar to a threshold mix network. If there are enough messages to be mixed, they start the \textbf{verifiable shuffling} as mentioned before. The incoming messages are in the form of $(m\cdot gpk_i^{r_1}, g^{r_1})$ where $r_1$ is an unknown random number and $i$ is the index of the chosen trusted forwarder, as the ElGamal encryption system is used in \protocol. To shuffle messages, the public key of the messages is required. However, each group member has a different group public key, and group members do not know which public key is the correct one. 

We introduce the expansion-and-shuffling procedure to solve this issue. We assume that all parties can be the trusted forwarder for every message and shuffle messages with every group public key.
We denote the incoming message list kept by the member $j$ as list $j$. For messages in the list $j$, member $j$ shuffles them using $gpk_j$. Concretely, a group member samples a new random number $r_2$, calculates $(m\cdot gpk_i^{r_1}\cdot {gpk_j}^{r_2}, g^{r_1}\cdot g^{r_2})$ for every $m$ on the list, and shuffles the ciphertexts obtained with a random permutation. Hence, we can obtain a valid encrypted message $(m\cdot {gpk_i}^{r_1+r_2}, g^{r_1+r_2})$ if the correct public key is used, i.e., $j=i$. Otherwise, we have a message $(m\cdot gpk_i^{r_1}\cdot {gpk_j}^{r_2}, g^{r_1}\cdot g^{r_2})$ that can only be decrypted to a random group element. Afterwards, the shuffled list is forwarded to the next member of the group (the next member is the member $(j+1) \ mod \ n$ if the index of the current member $j$). The group key index $j$ is also sent with the shuffled message so that the member $j+1$ knows the key used to shuffle the received messages. The member $j+1$ shuffles the received messages with $gpk_j$ and sends them to the member $j+2$. This process continues until member $j-1$, which means that every participant has shuffled the messages. For each shuffling, members generate a proof based on the verifiable shuffling algorithm proposed by Groth~\cite{verifiableShuffling}. Hence, they cannot shuffle messages with a wrong group public key. In practice, shuffling processes of different group keys can be parallelized to accelerate the shuffling of messages.

The remaining challenge is how to detect and remove malformed ciphertexts from the shuffled output. To address this, the sender transmits two related ciphertexts with each message: $m1 = (m\cdot {gpk_i}^{r_1}, g^{r_1})$ and $m2 = (m^{-1}\cdot {gpk_i}^{r'_1}, g^{r'_1})$. These two ciphertexts are processed together throughout the shuffling procedure and enable group members to verify message correctness without revealing the content of a message.
Both $m_1$ and $m_2$ are shuffled using the same expansion-and-shuffling procedure described above. However, it is essential that the two ciphertext lists are shuffled under the same hidden permutation so that corresponding pairs remain aligned. The verifiable shuffle scheme of Groth~\cite{verifiableShuffling} guarantees that for a list of inputs $x_1,\ldots,x_k$ and outputs $y_1,\ldots,y_k$, there exists a permutation $\pi$ and randomizers such that $y_{\pi(i)} = x_i \cdot g^{r_i}$.
In our construction, we maintain two lists $(m_{1,1},\ldots,m_{1,k})$ and $(m_{2,1},\ldots,m_{2,k})$, which are shuffled independently to produce $(y_{1,1},\ldots,y_{1,k})$ and $(y_{2,1},\ldots,y_{2,k})$. To prove that the same permutation was applied to both lists, group members additionally prove that the products of $(y_{1,1} \cdot y_{2,1}, \ldots, y_{1,k} \cdot y_{2,k})
$ form a valid shuffle of $(m_{1,1} \cdot m_{2,1}, \ldots, m_{1,k} \cdot m_{2,k})$. A valid proof guarantees that each $y_{1,j}$ is paired with its correct counterpart $y_{2,j}$, preserving the correspondence between the two ciphertext components.

After shuffling, \textbf{decryption and check} takes place. we multiply $m1_j$ and $m2_j$ for all $j$. If $m1_j$ and $m2_j$ are shuffled using the correct group public key, they should be the form of $(m\cdot {gpk_i}^{r}, g^{r})$ and $(m^{-1}\cdot {gpk_i}^{r'}, g^{r'})$ where $r$ and $r'$ are two secret random numbers. Thus, the multiplication equals $({gpk_i}^{r+r'},g^{r+r'})$, that is, the encryption of $1$. If a wrong group public key is used, messages are in the form of $(m\cdot {gpk_i}^{r_1}\cdot {gpk_j}^{r_2}, g^{r_1+r_2})$ and $(m^{-1}\cdot {gpk_i}^{r_1'}\cdot {gpk_k}^{r_2'}, g^{r_1'+r_2'})$ where $gpk_i$ is the group public key used to encrypt the message, $gpk_j$ is the group public key used to shuffle the message, $r_1$ and $r_1'$ are random numbers determined by the previous group, $r_2$, $r_2'$ are  random numbers chosen by current group members. Then, the multiplication of two messages equals $({gpk_i}^{r_1+r_1'}\cdot{gpk_j}^{r_2+r_2'} ,g^{r_1+r_1'+r_2+r_2'})$. Hence, group members are able to check whether a message is in the correct form by decrypting the multiplication of two linked messages. If the decrypted result equals $1$, the integrity check of the message passes, and the owner of the corresponding public key is the trusted forwarder of this message. Otherwise, the wrong public key is used to shuffle this message, and the corresponding message should be dropped. Similarly to Atom~\cite{Atom}, a proof showing that a member has decrypted the message using the correct secret key is generated during the decryption of the messages to prevent malicious behaviors. In conclusion, group members know there is a message for a specific group member but are unable to link it to the corresponding incoming message. 

It is possible that the number of ciphertext pairs that decrypt to $1$ is less than the number of messages in the incoming message list. This may occur either because some group members misbehave during the decryption process or because some messages are malformed. In the former case, misbehavior can be detected since participants generate proofs of correct decryption, which can be verified by other group members to identify misbehaving parties. In the latter case, malformed messages are simply discarded without further consequences.

After decryption and check, the trusted forwarder of the current group chooses the next trusted forwarder in a neighboring group and initializes the \textbf{anonymous public key re-encryption}. To hide the identity of the next trusted forwarder, it asks group members to re-encrypt the ciphertext under an anonymous group public key $apk$. $apk$ is the product of the next trusted forwarder's group public key $gpk'_i$ and a randomly sampled public key $rpk$, chosen by the current trusted forwarder. Although every message is end-to-end encrypted, adversaries are able to link a received message to a specific sender based its cipher text.
Therefore, we leverage a zero-knowledge proof~\cite{zk-snark} that shows $\exists (i,rpk,rsk) \ where\  gpk_i'\cdot rpk = apk \land rpk = g^{rsk}$ to prevent the trusted forwarder from embedding its own public key into $apk$ to learn the content of a message. After re-encryption, the ciphertext pair is in the form of $(m\cdot apk^{r_1},g^{r_1})||(m^{-1}\cdot apk^{r_2},g^{r_2})$.

In the last step of the message processing (\textbf{decryption with the anonymous secret key}), the trusted forwarder removes the blinding factor $rpk$ in the ciphertext by computing $(g^{r_1})^{rsk} = rpk^{r_1}$ and dividing it out from the first component, yielding $(m \cdot {gpk'}_i^{r_1}, g^{r_1})$. Similarly, $(m^{-1}\cdot {gpk'_i}^{r_2},g^{r_2})$ is calculated for the inverse message. 
The trusted forwarder must generate a zero-knowledge proof that it knows the secret key $rsk$ corresponding to the $rpk$ used to generate $apk$, and that the transformation from $apk$-encryption to $gpk'_i$-encryption is performed correctly using this same $rsk$. This ensures that the ciphertext was not manipulated using incorrect keys. After this step, the resulting ciphertext pair is correctly encrypted under the group public key of the next trusted forwarder.

\section{Security}
\label{sec:security}

In \protocol, the trusted party routes the message to the next group. For parties who accidentally trust the adversary, the adversary is able to route the message to a malicious group and decrypt the message without any honest participants. Such a behavior cannot be prevented in re-encryption based mixing protocols if the forwarder is malicious. Therefore, we only prove the security of \protocol for honest groups similar to the security analysis of Atom \cite{Atom}. By an honest group, we mean a group such that either i) the forwarder selected as the trusted party is honest or ii) there are at least $n-n_{min}+1$ honest parties ensuring that there is at least one honest participant during the shuffling and decryption of messages. Compared with Atom, \protocol does not rely on the "many-trust assumption"~\cite{Atom} if the trusted party is honest and provides a similar security level when the trusted party is an adversary. 

We prove two security properties for \protocol: Message unlinkability ensures that outgoing messages and incoming messages of a group are unlinkable, whereas message integrity ensures that messages cannot be modified or forged, not even by the trusted forwarder. In our proofs, we assume that the adversary controls the trusted forwarder. This represents the strongest adversarial setting and avoids redundant arguments for cases in which the adversary controls at most $n_{min}$ group members but not the trusted forwarder, since in those scenarios the adversary obtains less information than when the trusted forwarder is compromised.

In the security games of our proofs, we simulate the shuffling process under the assumption that at least one participant in each group is honest. We use a simulator to simulate the shuffling process with at least one honest participant based on the random oracle model. Concretely, we use the shuffling scheme proposed by Groth~\cite{verifiableShuffling} that provides a Special Honest Verifier Zero-Knowledge (SHVZK) argument of knowledge for correctness shuffling. Informally, this means that a prover can convince an honest verifier that a set of ciphertexts is a valid re-encryption and permutation of an input set, without revealing the underlying permutation or plaintexts, while guaranteeing that any prover that succeeds must know such a permutation and the corresponding randomness. By the soundness of the SHVZK argument of knowledge for correct shuffling, no adversary can produce an incorrect shuffle that is accepted, and by the zero-knowledge property, no information about the permutation is leaked. Therefore, in our proofs we abstract away the concrete multi-party shuffling procedure and replace it with a simulator that models the shuffle via a random oracle. Specifically, the oracle selects a random permutation $\pi$ and re-randomizes all incoming messages so $\pi$ is hidden for the adversary.  Regarding the decryption process after shuffling, it is obvious that decryption of encrypted $1$ does not leak any information about the underlying plaintext message or the applied permutation. Therefore, we abstract away the concrete decryption procedure and give the adversary with the correctly shuffled ciphertexts in our security games. It simulates the outcome of the decryption phase while not weakening the adversary. 

\subsection{Integrity of messages}
Integrity implies that an adversary cannot alter the content of incoming messages during the execution of a protocol. To prove the integrity of \protocol, we simulate \protocol with a game $G_{MIntegrity}$ as described in Algorithm~\ref{game:Integrity} and show that a PPT adversary $A$ is not able to tamper with an incoming message with a non-negligible probability. 

In Algorithm~\ref{game:Integrity}, $k$ messages, trust public keys, and trust secret keys are generated using random oracle $O^{ram}$ at first. It implies that messages are sent from parties that are not controlled by the adversary, so $A$ does not know the trust secret key. 
The case when some messages are sent from parties controlled by $A$, it is equivalent to the case where there are fewer incoming messages, i.e., lower $k$, which is still covered by Algorithm~\ref{game:Integrity}. By the correctness of the underlying key generation protocol~\cite{dkgprotocol}, each participant should hold a public/secret key pair for each group public key after the key generation process. In our security proofs, we abstract away the concrete multi-party key generation protocol and instead simulate its outcome as follows: $A$ and $O^{ram}$ each generates $1$ public key for each group public key. It simulates the case where there is at least one honest party during the construction of $k$ group public keys. Afterwards, we use $O^{ram}$ to encrypt and sign generated messages and forward the signed ciphertext to $A$. During the encryption of a message, a random permutation $p$ for the $k$ group public keys is generated by $O^{ram}$ to ensure $A$ does not know which group public key is used to encrypt a message. Similarly, a random permutation $\pi$ is also generated for the incoming message. In Line~\ref{line:shuffling} of Algorithm~\ref{game:Integrity}, we let $O^{ram}$ permute the generated messages based on $\pi$ and re-randomize them to simulate the shuffling of incoming messages with at least one honest participant. 
Afterwards, $A$ chooses a shuffled message to be attacked. Next, $A$ chooses the group public keys of the next group and also the anonymous key $apk$ and the corresponding proof of $apk$. If the proof cannot be verified correctly, $A$ loses the game. If the proof is correct, $O^{ram}$ outputs the re-encrypted message with the selected $apk$ to $A$ to simulate a correct execution of the re-encryption algorithm. Here, we assume that a verifiable re-encryption scheme that satisfies correctness like the Chaum-Pedersen proof \cite{ChaumProof92} is used, so $A$ cannot forge the message in this step. If the proof is not correct, the honest party notices the existence of adversaries and stops the execution of the protocol. Afterwards, $A$ needs to replace the anonymous key with the actual group public key of the next proof. $A$ outputs the message encrypted by the next group public key and the proof that shows the re-encryption is done correctly. If the output message is not the original message and the proof can be verified successfully, $A$ wins the game.

\begin{definition}[Integrity] A protocol satisfies integrity if any Probabilistic Polynomial-Time (PPT) adversary $A$ can not win message integrity  game $G_{MIntegrity}$ in Algorithm~\ref{game:Integrity} with a non-negligible probability. 
\end{definition}

\begin{theorem}
If the shuffling scheme provides an SHVZK argument of knowledge for correctness of the shuffling, the re-encryption scheme satisfies correctness, the underlying zero-knowledge proof protocol satisfies correctness and soundness, and there is at least one honest participant during the execution of \protocol, \protocol satisfies message integrity. 
\end{theorem}

\begin{proof}
In $G_{MIntegrity}$, after selecting the message to be attacked, $A$ chooses an anonymous public key $apk$ with a proof showing that $apk= gpk'_i\cdot rpk$ and $A$ know the corresponding secret key $rsk$ for $rpk$ where $i$ is an arbitrary index for the generated group public keys of $A$.  Then, $A$ receives the re-encrypted message and needs to output the corresponding outgoing message with a proof for the selected incoming message in Line~\ref{line:outgoingmessage} of Algorithm~\ref{game:Integrity}. By the soundness and correctness of the proof system underlying $Verify2$, acceptance implies the existence of witnesses such that $m_A = m_{b}\cdot apk^{r''_{1,b}} / {g^{r_{1,b}''\cdot rsk}}$ and $g^{r_{1,b}''\cdot rsk} = rpk^{r_{1,b}''}$ hold. So, $m_A = m_{b}\cdot apk^{r''_{1,b}} / rpk^{r_{1,b}''}$ holds. Because the proof also requires that $apk$ equals $gpk'_i\cdot rpk$, $m_A = m_{b}\cdot {gpk'}_{i}^{r''_{1,b}}$ holds ultimately. Similarly, $m'_A=m^{-1}_{b}\cdot {gpk'}_{i}^{r''_{2,b}}$. Therefore, any accepting output must correspond to a correct re-encryption of the challenged message. Hence, no PPT adversary $A$ can win $G_{MIntegrity}$ with a non-negligible probability.
\end{proof}

\begin{algorithm}
\caption{$G_{MIntegrity}$}\label{alg:MIntegrity}
\label{game:Integrity}
\begin{algorithmic}[1]
\State $ {(m_i,TPK_i,TSK_i)_{i\in1..k}}\leftarrow O^{ram}$
\State $A \rightarrow(pk_i)_{i \in 1..k}$
\State $ (pk^{honest}_i)_{i \in 1..k} \leftarrow O^{ram}$
\State $A\leftarrow(gpk_i:=pk_i\cdot pk^{honest}_i)_{i \in 1..k} \leftarrow O^{ram}$
\State $A \leftarrow ({m_i\cdot gpk_{p_i}^{r_{1,i}},g^{r_{1,i}}},{m_i^{-1}\cdot gpk_{p_i}^{r_{2,i}},g^{r_{2,i}}},sig_i)_{i\in1..k} \leftarrow O^{ram}$
\State $A \leftarrow $ $($ ${m_{\pi_i}\cdot gpk_{p_{\pi_i}}^{r_{1,{\pi_i}}}\cdot gpk_{j}^{r'_{1,ik+j}},g^{r_{1,\pi_i}+r'_{1,ik+j}}}, 
m_{\pi_i}^{-1}\cdot gpk_{p_{\pi_i}}^{r_{2,{\pi_i}}}\cdot gpk_{j}^{r'_{2,ik+j}},$ $g^{r_{2,{\pi_i}}+r'_{2,ik+j}}$ $)_{i\in1..k,j\in1..k}$ $\leftarrow O^{ram}$
\State $ let \ r''_{1,\pi_i}  = r_{1,{\pi_i}}+r^{'}_{1,k{\pi_i}+p_{\pi_i}}\ , \ r''_{2,\pi_i}  = r_{2,{\pi_i}}+r^{'}_{2,k{\pi_i}+p_{\pi_i}} $
\label{line:shuffling}
\label{line:chooseMessage}
\State $A \leftarrow($ ${m_{\pi_i}\cdot gpk_{p_{\pi_i}}^{r''_{1,{\pi_i}}},g^{r''_{1,\pi_i}}}$$,m_{\pi_i}^{-1}\cdot gpk_{p_{\pi_i}}^{r''_{2,{\pi_i}}},$ $g^{r''_{2,{\pi_i}}}$$)_{i\in1..k}$
\State $b \leftarrow A$
\State $A \rightarrow((gpk'_i)_{i \in 1..k} \ ,apk,\delta)$
\State $\delta: \text{A knows }rsk,rpk \text{ such that } \exists i: gpk'_{i}\cdot rpk = apk \land rpk = g^{rsk}$
\IF{$!Verify(apk,(gpk'_i)_{i \in 1..k},\delta)$}
\State$return\ 0$
\ENDIF{}
%\State $A\leftarrow({gpk_{p_{b}}^{r_{1,{b}}}\cdot gpk_{b}^{-r_{1,b}}},{gpk_{p_{b}}^{r_{2,{b}}}\cdot gpk_{b}  ^{-r_{2,b}}})\leftarrow O^{ram}$
\State $A\leftarrow ({m_{b}\cdot apk^{r''_{1,b}},g^{r''_{1,b}}},{m_{b}^{-1}\cdot apk^{r''_{2,b}},g^{r''_{2,b}}})\leftarrow O^{ram}$
\label{line:outgoingmessage}
\State $A\rightarrow ({m_A, g^{r''_{1,b}}},{m_A',g^{r''_{2,b}}},\delta')$
\State $\delta': \text{A knows }rsk,rpk \text{ such that } \exists i: gpk'_{i}\cdot rpk = apk \land rpk = g^{rsk} \land m_{b}\cdot apk^{r''_{1,b}} / m_A = {g^{r_{1,b}''\cdot rsk}} \land m_{b}^{-1}\cdot apk^{r''_{2,b}}/ m_A' = {g^{r_{2,b}''\cdot rsk}}$
\IF{$Verify2($ ${m_A, g^{r''_{1,b}}},{m_A',g^{r''_{2,b}}},{m_{b}\cdot apk^{r''_{1,b}},g^{r''_{1,b}}},$$m_{b}^{-1}\cdot apk^{r''_{2,b}},g^{r''_{2,b}},apk,\delta')\land m_A\neq m_{b}\cdot {gpk'_{i}}^{r''_{1,b}}$}
\State $return \ 1$
\ENDIF{}
\State $return \ 0$
\end{algorithmic}
\end{algorithm}

\subsection{Unlinkability of messages}
The unlinkability game $G_{MUnlinkability}$ described in Algorithm~\ref{alg:MUnlinkability} is almost the same as $G_{MIntegrity}$ except that the goal of $A$ is to link an outgoing message to an incoming message now. Instead of letting $A$ choose the index of a message by itself, $O^{ram}$ chooses the incoming messages with a randomly selected index $b$ in $G_{MUnlinkability}$, so $A$ can guess the index $b$ later. Afterwards, the same re-encryption and verification process as in $G_{MIntegrity}$ occurs. In Line~\ref{line:guessB} of Algorithm~\ref{alg:MUnlinkability}, however, $A$ needs to output the guessed index $b'$ instead of outputting a modified message. If the guessed $b'$ is equal to the index selected by $O^{ram}$ before, $A$ wins $G_{MUnlinkability}$.  

\begin{definition}[Unlinkability]
A protocol satisfies unlinkability if any Probabilistic Polynomial-Time (PPT) adversary $A$ cannot win the message unlinkability game $G_{MUnlinkability}$ described in Algorithm~\ref{alg:MUnlinkability} with a probability higher than $1/k+negl$ where $negl$ is a negligible probability and $k$ is the number of messages.
\end{definition}

\begin{theorem}
If the Decisional Diffie–Hellman assumption holds, the shuffling scheme provides an SHVZK argument of knowledge for correctness of the shuffling, the re-encryption scheme satisfies correctness, the underlying zero-knowledge proof protocol satisfies correctness and soundness, and there is at least one honest participant during the execution of \protocol, \protocol satisfies message unlinkability. 
\end{theorem}

\begin{algorithm}
\caption{$G_{MUnlinkability}$}\label{alg:MUnlinkability}
\begin{algorithmic}[1]
\State $ {(m_i,TPK_i,TSK_i)_{i\in1..k}}\leftarrow O^{ram}$
\State $A \rightarrow(pk_i)_{i \in 1..k}$
\State $ (pk^{honest}_i)_{i \in 1..k} \leftarrow O^{ram}$
\State $A\leftarrow(gpk_i:=pk_i\cdot pk^{honest}_i)_{i \in 1..k} \leftarrow O^{ram}$
\State $A \leftarrow ({m_i\cdot gpk_{p_i}^{r_{1,i}},g^{r_{1,i}}},{m_i^{-1}\cdot gpk_{p_i}^{r_{2,i}},g^{r_{2,i}}},sig_i)_{i\in1..k} \leftarrow O^{ram}$
\label{line:sign}
\State $A \leftarrow $ $($ ${m_{\pi_i}\cdot gpk_{p_{\pi_i}}^{r_{1,{\pi_i}}}\cdot gpk_{j}^{r'_{1,ik+j}},g^{r_{1,\pi_i}+r'_{1,ik+j}}}, 
m_{\pi_i}^{-1}\cdot gpk_{p_{\pi_i}}^{r_{2,{\pi_i}}}\cdot gpk_{j}^{r'_{2,ik+j}},$ $g^{r_{2,{\pi_i}}+r'_{2,ik+j}}$ $)_{i\in1..k,j\in1..k}$ $\leftarrow O^{ram}$
\label{line:shufflecipher}
\State $ let \ r''_{1,\pi_i}  = r_{1,{\pi_i}}+r^{'}_{1,k{\pi_i}+p_{\pi_i}}\ , \ r''_{2,\pi_i}  = r_{2,{\pi_i}}+r^{'}_{2,k{\pi_i}+p_{\pi_i}} $
\State $A \leftarrow($ ${m_{\pi_i}\cdot gpk_{p_{\pi_i}}^{r''_{1,{\pi_i}}},g^{r''_{1,\pi_i}}}$$,m_{\pi_i}^{-1}\cdot gpk_{p_{\pi_i}}^{r''_{2,{\pi_i}}},$ $g^{r''_{2,{\pi_i}}}$$)_{i\in1..k}$
\label{line:correctshufflecipher}
\State $ b \leftarrow O^{ram}$
%\State $A \leftarrow($ ${m_{b}\cdot gpk_{p_{b}}^{r_{1,{b}}}\cdot gpk_{b}^{r_{1',b}},g^{r_{1,b}+r_{1',b}}},m_{b}^{-1}\cdot gpk_{p_{b}}^{r_{2,{b}}}\cdot gpk_{b}^{r_{2',b}},$ $g^{r_{2,{b}}+r'_{2,b}}$$)\leftarrow O^{ram} $
\label{line:generatenewgroupkey}
\State $A \rightarrow((gpk'_i)_{i \in 1..k} \ ,apk,\delta)$
\State $\delta: \text{A knows }rsk,rpk \text{ such that } \exists i: gpk'_{i}\cdot rpk = apk \land rpk = g^{rsk}$
\IF{$!Verify(apk,(gpk'_i)_{i \in 1..k},\delta)$}
\State$return\ 0$
\ENDIF{}
%\State $A\leftarrow({gpk_{p_{b}}^{r_{1,{b}}}\cdot gpk_{b}^{-r_{1,b}}},{gpk_{p_{b}}^{r_{2,{b}}}\cdot gpk_{b}  ^{-r_{2,b}}})\leftarrow O^{ram}$
\State $A\leftarrow ({m_{b}\cdot apk^{r''_{1,b}},g^{r''_{1,b}}},{m_{b}^{-1}\cdot apk^{r''_{2,b}},g^{r''_{2,b}}})\leftarrow O^{ram}$
\label{line:generateoutgoingmessage}
\State $A\rightarrow ({m_A, g^{r''_{1,b}}},{m_A',g^{r'_{2,b}}},\delta')$
\State $\delta': \text{A knows }rsk,rpk \text{ such that } \exists i: gpk'_{i}\cdot rpk = apk \land rpk = g^{rsk} \land m_{b}\cdot apk^{r''_{1,b}} / m_A = {g^{r_{1,b}''\cdot rsk}} \land m_{b}^{-1}\cdot apk^{r''_{2,b}}/ m_B = {g^{r_{2,b}''\cdot rsk}}$
\IF{$Verify2($ ${m_A, g^{r''_{1,b}}},{m_A',g^{r''_{2,b}}},{m_{b}\cdot apk^{r''_{1,b}},g^{r''_{1,b}}},$$m_{b}^{-1}\cdot apk^{r''_{2,b}},g^{r''_{2,c}},apk,\delta')\land m_A\neq m_{b}\cdot {gpk'_{i}}^{r''_{1,b}}$}
\label{line:guessB}
\State $A\rightarrow b'$
\IF {$b'==b$}
\State $return\ 1$
\ENDIF{}
\State $return \ 0$
\label{line:endogunlinkability}
\ENDIF{}
\end{algorithmic}
\end{algorithm}

\begin{algorithm}
\caption{$G_{MUnlinkability}^{2}$}
\label{game:unlinkability3}
\begin{algorithmic}[1]
\State $ {(m_i)_{i\in1..k}}\leftarrow O^{ram}$
\State $A \rightarrow(pk_i)_{i \in 1..k}$
\State $ (pk^{honest}_i)_{i \in 1..k} \leftarrow O^{ram}$
\State $A\leftarrow(gpk_i:=pk_i\cdot pk^{honest}_i)_{i \in 1..k} \leftarrow O^{ram}$
\State $A \leftarrow ({m_i\cdot gpk_{p_i}^{r_{1,i}},g^{r_{1,i}}},{m_i^{-1}\cdot gpk_{p_i}^{r_{2,i}},g^{r_{2,i}}},sig_i)_{i\in1..k} \leftarrow O^{ram}$
\State $A \leftarrow  ( {m_{\pi_i}\cdot {g^{x'_{i\cdot k+j}}},g^{r_{1,\pi_i}+r'_{1,ik+j}}}, 
m_{\pi_i}^{-1}\cdot  {g^{x'_{k^2+i\cdot k+j}}}, g^{r_{2,{\pi_i}}+r'_{2,ik+j}} )_{i\in1..k,j\in1..k}\leftarrow O^{ram}$
\label{randomShuffled}
\State $ let \ r''_{1,\pi_i}  = r_{1,{\pi_i}}+r^{'}_{1,k{\pi_i}+p_{\pi_i}}\ , \ r''_{2,\pi_i}  = r_{2,{\pi_i}}+r^{'}_{2,k{\pi_i}+p_{\pi_i}} $
\State $A \leftarrow($ ${m_{\pi_i}\cdot g^{x''_{i}},g^{r''_{1,\pi_i}}}$$,m_{\pi_i}^{-1}\cdot g^{x''_{i+k}},$ $g^{r''_{2,{\pi_i}}}$$)_{i\in1..k}$
\label{randomCorrectShuffled}
\State $ b \leftarrow O^{ram}$
%\State $A \leftarrow($ ${m_{b}\cdot gpk_{p_{b}}^{r''_{1,{b}}},g^{r''_{1,b}}},m_{b}^{-1}\cdot gpk_{p_{b}}^{r''_{2,{b}}},$ $g^{r''_{2,{b}}}$$)\leftarrow O^{ram} $
%\State $A \leftarrow($ ${m_{b}\cdot gpk_{p_{b}}^{r_{1,{b}}}\cdot gpk_{b}^{r_{1',b}},g^{r_{1,b}+r_{1',b}}},m_{b}^{-1}\cdot gpk_{p_{b}}^{r_{2,{b}}}\cdot gpk_{b}^{r_{2',b}},$ $g^{r_{2,{b}}+r'_{2,b}}$$)\leftarrow O^{ram} $
\State $A \rightarrow((gpk'_i)_{i \in 1..k} \ ,apk,\delta)$
\State $\delta: \text{A knows }rsk,rpk \text{ such that } \exists i: gpk'_{i}\cdot rpk = apk \land rpk = g^{rsk}$
\IF{$!Verify(apk,(gpk'_i)_{i \in 1..k},\delta)$}
\State$return\ 0$
\ENDIF{}
%\State $A\leftarrow({gpk_{p_{b}}^{r_{1,{b}}}\cdot gpk_{b}^{-r_{1,b}}},{gpk_{p_{b}}^{r_{2,{b}}}\cdot gpk_{b}  ^{-r_{2,b}}})\leftarrow O^{ram}$
\State $A\leftarrow ({m_{b}\cdot apk^{r''_{1,b}},g^{r''_{1,b}}},{m_{b}^{-1}\cdot apk^{r''_{2,b}},g^{r''_{2,b}}})\leftarrow O^{ram}$
\label{apkcipher}
\State $A\rightarrow b'$
\IF {$b'==b$}
\State $return\ 1$
\ENDIF{}
\end{algorithmic}
\end{algorithm}

\begin{proof}
We prove the unlinkability as a sequence of games, and remove steps that do not influence $A$'s winning rate gradually.

After the shuffling of messages, $A$ obtains a uniformly randomly permutation of re-randomized ciphertexts as described in Line~\ref{line:shufflecipher} of Algorithm~\ref{alg:MUnlinkability}. To show the re-randomized ciphertexts and the re-encryption of the chosen ciphertext are not helpful for $A$ to win the game, we define a sequence of hybrid games $H_0,H_1\cdots H_{2k^2+2k}$ where $H_0$ is $G_{MUnlinkability}$. In the first $2k^2$ games, the Diffie-Hellman elements of the shuffled ciphertexts in Line~\ref{line:shufflecipher} of Algorithm~\ref{alg:MUnlinkability} are replaced by $2k^2$ uniformly random group elements $g^{x'_i}$. We need $2k^2$ games because the shuffling of $k$ pairs of encrypted messages and inverse messages under $k$ group public keys generates $2k^2$ ciphertexts. Similarly, $2k$ Diffie–Hellman elements of correctly shuffled ciphertexts in Line~\ref{line:correctshufflecipher} of Algorithm~\ref{alg:MUnlinkability} are also replaced by $2k$ uniformly random group elements. For ciphertexts related to $apk$, we are not able to replace them with random group elements because $apk$ and $gpk'$ are chosen by $A$, and it can decrypt the ciphertexts to check whether it is from a random group element or a valid message. 
Finally, $G_{MUnlinkability}^{2}$ is obtained where all ciphers after shuffling except the re-encryption of the chosen message are replaced with random elements. 

In $G_{MUnlinkability}^{2}$, ciphertexts after shuffling received by $A$ are of the form $(m_i\cdot g^{x_i},g^r_i)$ and $({m_i}^{-1}\cdot g^{y_{i}},g^{r'_i})$ where all exponents are chosen independently and uniformly at random. For any fixed message $m_i$, the distribution of $m_i\cdot g^{x_i}$ is uniform. Therefore, ciphertexts in Line~\ref{randomShuffled} and Line~\ref{randomCorrectShuffled} of Algorithm~\ref{game:unlinkability3} are independent of the chosen index $b$ from the view of $A$ in $G_{MUnlinkability}^{2}$. 
For the ciphertexts in Line~\ref{apkcipher} of Algorithm~\ref{game:unlinkability3}, the only component that depends on the chosen index $b$ is the plaintext message $m_b$. We can assume that $A$ obtains $m_b$ by decrypting the ciphertexts using the secret key corresponding to $apk$. However, knowing $m_b$ does not reveal the corresponding incoming message, since the incoming ciphertexts are distributed as independent ElGamal encryptions of the plaintexts under uniformly random randomness. Consequently, conditioned on $m_b$, every incoming ciphertext is also equally likely to correspond to it.
Hence, $Pr[A \ wins \ G_{MUnlinkability}^{2}] = 1/k$. Meanwhile, if an adversary can distinguish $H_j$ from $H_{j+1}$, it implies that this adversary is capable of distinguishing whether the output ciphertext contains $gpk_{i}^{r_i}$ or a random number $g^{x_i}$, which can be used to break the DDH assumption. Hence, $|Pr[A(H_j)=1]-Pr[A(H_{j+1}=1)]|\leq Adv_{DDH}$ where $Pr[A(H_j)=1]$ is the probability that $A$ wins $H_j$.
Combining this with the win rate of $G_{MUnlinkability}^{2}$, we obtain 
\begin{gather*}
Pr[A \ wins \ G_{MUnlinkability}] \\\leq 1/k+ \sum_{j=0}^{2k^2+2k} |Pr[A(H_j)=1]-Pr[A(H_{j+1}=1)]| \\
\\ \leq 1/k+ (2k^2+2k)\cdot Adv_{DDH}
\end{gather*}
As $Adv_{DDH}$ is negligible, $(2k^2+2k)\cdot Adv_{DDH}$ is negligible. In conclusion, $A$ cannot win $G_{MUnlinkability}$ with a probability non-negligibly higher than $1/k$ and \protocol satisfies unlinkability.
\end{proof}

\section{Evaluation}
We evaluate \protocol both as a proof-of-concept Android implementation on real phones and as a simulation to scale to hundreds of users, allowing us to characterize the level of anonymity.

\subsection{Proof-of-Concept implementation}
We implement \protocol for a single group and evaluate the performance of \protocol with a different number of participants. As different groups are able to process messages independently, the results of our implementation can be used to estimate the throughput of \protocol. 

Our implementation~\cite{myImplementation} is based on the WIFI-aware feature of Android phones. In our implementation, users find nearby devices through WIFI-aware discovery, which provides information about other devices, including device ID, application version, and starting timestamp.  After receiving discovery messages from all members of the group, an index is assigned to each member deterministically based on their device ID and timestamp. Afterwards, each pair of participants establishes a direct connection and the group key construction is executed as described in Section~\ref{sec:protocol}. In our implementation, we use the P-256 curve based ECC encryption scheme to realize 128-bit security~\cite{bookforp256}. We let a specific member generate messages and broadcast the generated messages to simulate messages collected from neighbors. Afterwards, each party executes \protocol as described in Section~\ref{sec:protocol}. We also implement the Atom shuffling algorithm~\cite{Atom} using the same framework for comparison. 
%Afterwards, each group member uses its group key to shuffle the message pool and pass the shuffled messages to the $(index+1) \ mod \ n$ member where $index$ is the index of the current member and $n$ is the size of the current group. When shuffled messages are received by the initial shuffler again, it means that every member has finished the shuffling process. Then, the initial shuffler multiplies two parts of each message and initializes the decryption process for the multiplication. After every group member has performed the decryption of the messages, there will be some results equal to $1$. For the corresponding messages of those $1s$, the initial shuffler generates an anonymous key and starts the re-encryption process for the message using the generated anonymous key. After the re-encryption process, the member removes the anonymous factor and forwards the messages to a corresponding group. In this implementation, every step of the message processing in \protocol is implemented. 

In our evaluation, we vary the number of users in a group and measure the time required to mix $100$ messages and compare \protocol with Atom. We keep the number of messages fixed because our simulation shows that a higher number of messages in the pool only has a minor influence on the anonymity of messages. The time consumption of each step is recorded locally by each device and the final result is the average of results in all devices. For the key exchange between participants, it can be completed in a second, and we ignore it in our result. We use five smartphones in our evaluation: two Samsung S25 devices (Snapdragon 8 Elite, Android 15) and three Samsung S21 devices (Snapdragon 888, Android 13). Experiments with one or two devices are conducted using only S25 phones. For experiments involving more devices, we incrementally add one S21 phone each time the number of participants increases. Those phones are placed on the same table. 

\begin{figure}[h]
\centering
\includegraphics[width=0.35\textwidth]{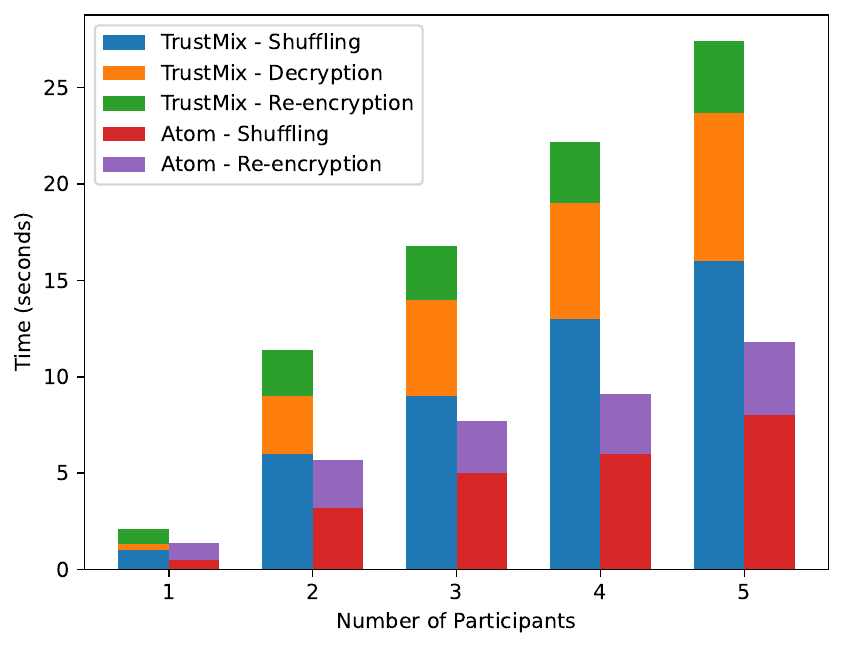}
\caption{Time consumption to mix 100 messages with different number of participants}
\label{fig:demoresult}
\end{figure}

As shown in Figure~\ref{fig:demoresult}, the time required for only one participant is low as no communication is involved, whereas for two or more participants, the required time increases linearly, as participants sequentially apply the shuffling, decryption, and re-encryption operations to a set of messages. For the comparison of \protocol and Atom, we find that the shuffling time of Atom is about half that of \protocol, which is due to the fact that \protocol processes two messages for every message, the encrypted message itself and its inverse. Shuffling of messages with multiple group public keys in a group does not significantly influence the speed of \protocol because different participants  shuffle messages using different group public keys in parallel. The same re-encryption process is performed by both \protocol and Atom, so they have a similar re-encryption time. Overall, the throughput of \protocol is approximately 45\% that of Atom. However, \protocol can deal with more misbehavior than Atom at the price of efficiency. For the increased number of messages, the processing time will increase linearly considering the time to shuffle, decrypt, re-encrypt, and exchange messages increase linearly as the number of messages increase. 
Although the throughput of \protocol for a 5-user group is only about 3.6 messages per second, it still shows that \protocol can be used in scenarios where messages are sent with low frequency and high anonymity is required. `

\subsection{Anonymity}
For the evaluation of anonymity, we simulate groups as nodes in a square system where each group has four neighbors on the top, bottom, left, and right sides. In practice, the number of neighboring groups depends on the density of devices, physical constraints, and the group formation strategy. To focus on the intrinsic anonymity properties of \protocol and to avoid confounding effects of network topologies, we only evaluate \protocol in a square topology.

We let each group generate $8$ messages per second and run an exiting mix network simulator~\cite{mixim} to evaluate the anonymity of the generated messages over 60 seconds. We run our simulation with 100 groups in a $10*10$ square 10 times and calculate the mean of obtained results. To simulate the message shuffling of \protocol, we let the simulator use the threshold mix strategy, that means a node shuffles messages in the message pool only when the number of messages reaches a defined threshold. We evaluate the anonymity of \protocol varying message pool size, number of mix nodes that every message goes through, and the number of corrupted groups. Specifically, for each output message, Mixim provides the probability distribution over possible senders, from which we compute entropy~\cite{serjantov2002towards} and report the average across all messages.
%Entropy~\cite{serjantov2002towards}, being the common metric to evaluate anonymity, is used to compare the privacy provided in the different settings, averaged over all messages. 

\begin{figure}[h]
\centering
\includegraphics[width=0.35\textwidth]{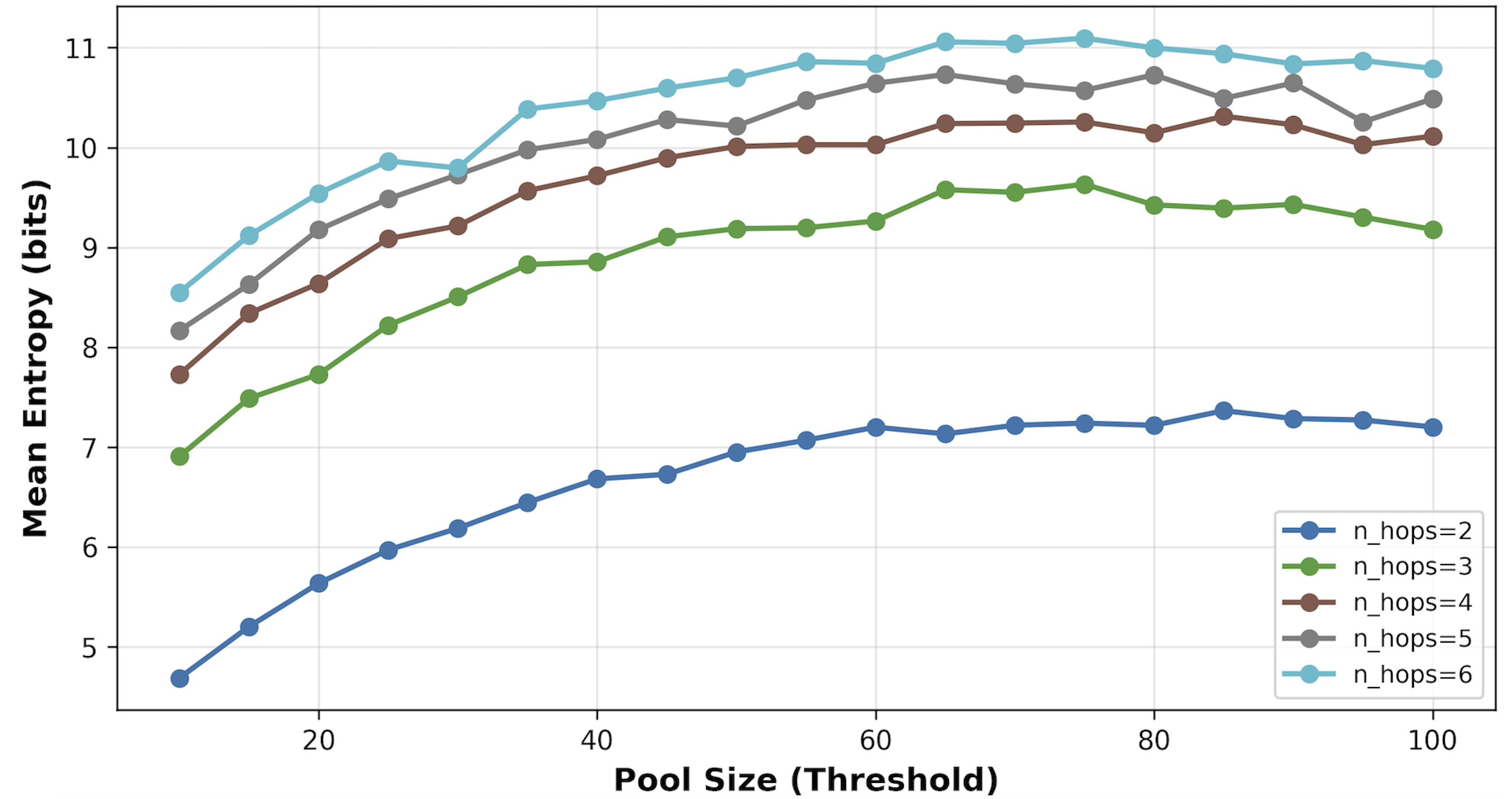}
\caption{Anonymity entropy of \protocol with a varying pool size and a different number of mixes}
\label{fig:differentPoolSize}
\end{figure}

Figure~\ref{fig:differentPoolSize} illustrates that the mixing significantly improves the anonymity of the message considering that forwarding messages without mixing results in an anonymity entropy equal to $1$, because the adversary is able to trace each message back to its sender based on traffic patterns. As the pool size increases, so does the entropy, as there are more messages that a message can be confused with. However, the increase tapers off at around a pool size of $60$, as the fixed number of messages per group during a period limits the anonymity. In addition, increasing the number of mixing hops further improves anonymity, since messages are shuffled across more groups with a larger set of messages. However, the benefit diminishes as the number of hops increases, indicating decreasing returns from a high number of  hops. 

\begin{figure}[h]
\centering
\includegraphics[width=7cm, height=5cm]{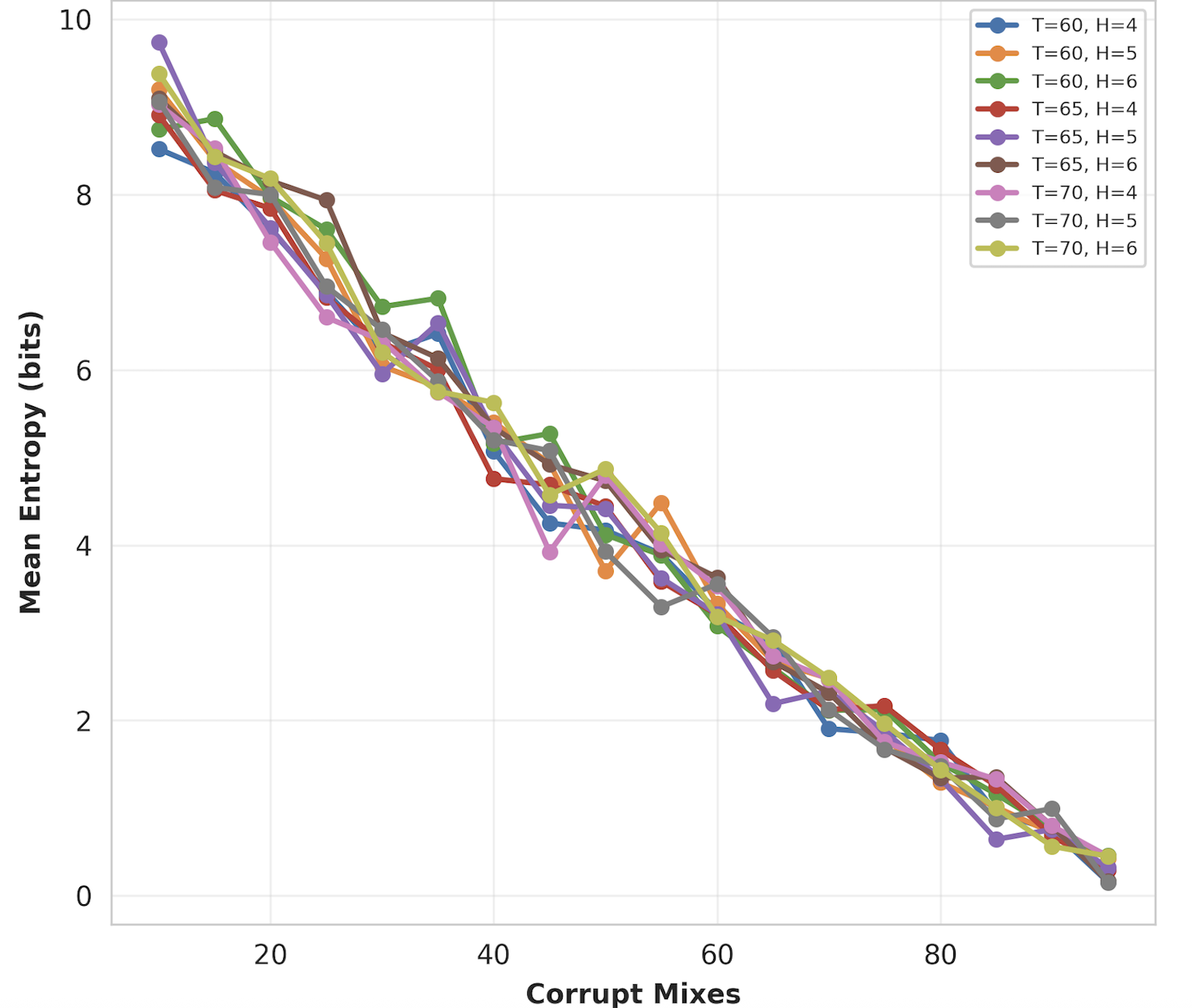}
\caption{Anonymity entropy of \protocol under different adversary ratios}
\label{fig:differentAdversary}
\end{figure}

In simulations with varying numbers of corrupted groups, we assume that the adversary is capable of decrypting and observing the plaintexts of all messages forwarded to corrupted groups. By a corrupted group, we mean that every member inside the group is controlled by the adversary. Consequently, once a message is forwarded to a corrupted group, subsequent mixing does not increase its anonymity entropy since the adversary can forward all incoming messages to another corrupted group.
Figure~\ref{fig:differentAdversary} shows that anonymity entropy decreases approximately linearly as the ratio of adversaries increases. Although the effectiveness of \protocol is limited when a large fraction of the network is controlled by adversaries, the results still indicate that it can provide a high level of anonymity when the majority of participants are honest.

\section{Conclusion}
In conclusion, \protocol realizes a significant anonymity increase in the square topology and even with the existence of active adversaries in the network. Meanwhile, our Proof-of-Concept implementation shows the feasibility of \protocol with real devices. Although the improvement of the anonimity in practice will depend on the underlying network topology and node connectivity, our results still show that \protocol is effective in enlarging the anonymity sets through mixing and can be expected increase the message anonymity in different topologies.

%%
%% The next two lines define the bibliography style to be used, and
%% the bibliography file.
\bibliographystyle{plain}
\bibliography{sample-base}

\end{document}